\documentclass[pra, twocolumn, amsmath, amssymb, superscriptaddress, nofootinbib]{revtex4}

\usepackage[colorlinks,linkcolor=red,citecolor=blue,urlcolor=blue]{hyperref}
\usepackage{graphicx}
\usepackage{amsmath}
\usepackage{amsthm}
\usepackage{amssymb}
\usepackage{amsbsy}
\usepackage{bbm}
\usepackage[usenames]{color}
\usepackage{bm}
\usepackage{flafter}
\usepackage{mathrsfs}

\newcommand{\ket}[1]{\left|#1\right\rangle}
\newcommand{\bra}[1]{\left\langle #1 \right|}

\newcommand{\uket}[1]{|#1\rangle}
\newcommand{\ubra}[1]{\langle #1 |}

\newcommand{\op}[1]{\hat{#1}}
\newcommand{\proj}[1]{\mathbbm{#1}} 
\newcommand{\spc}[1]{\mathcal{#1}}

\definecolor{DarkRed}{rgb}{0.7,0,0}

 \newtheorem{theorem}{Theorem}

\DeclareMathOperator{\Tr}{Tr}

\newcounter{ComntCntr} 

\newcommand{\pay}[1]{\mathfrak{p}^*(#1)}
\newcommand{\payr}[1]{\mathfrak{q}^*(#1)}

\begin{document}

\title{Quantifying entanglement when measurements are imperfect or restricted}

\author{Sebastian Meznaric}
\affiliation{Clarendon Laboratory, University of Oxford, Parks Road, Oxford OX1 3PU, United Kingdom}

\begin{abstract}
Motivated by the increasing ability of experimentalists to perform detector tomography, we consider how to incorporate the imperfections and restrictions of available measurements directly into the quantification of entanglement. Exploiting the idea that the recently introduced semiquantum nonlocal games as the gauge of the amount of entanglement in a state, we define an effective entanglement functional giving us effective entanglement when the measurement operators one has at their disposal are restricted by either fundamental considerations, such as superselection rules, or practical inability to conduct precise measurements. We show that effective entanglement is always reduced by restricting measurements. We define effective entanglement as the least amount of entanglement necessary to play all semiquantum nonlocal games at least as well with unrestricted measurements as with the more entangled original state and restricted measurements. We show that simple linear relationship between effective and conventional G-concurrence, generalization of concurrence, can be obtained when completely positive maps are used to describe measurement restrictions. We consider how typical measurement errors like photon loss and phase damping degrade effective entanglement in quantum optical experiments, as well as the impact of mass superselection rules on entanglement of massive indistinguishable particles. The flexibility of the effective entanglement formalism allows this single-particle entanglement to be calculated in the presence of a local Bose Einstein condensate reference frame with varying phase uncertainty, thereby interpolating between the complete breaking or strict application of the super-selection rule. 
\end{abstract}

\maketitle

\section{Introduction}
The act of measuring is not only the foundation for connecting quantum phenomena to classical events \cite{Braginski92}, it is also a key feature in emerging quantum technologies ranging from Heisenberg limited metrology to one-way quantum computing. Complementing this there has been remarkable progress recently \cite{Lundeen09, Feito09} in obtaining a complete characterisation of a quantum optical experiment by performing tomography not only of the input state and dynamics of an apparatus, but also of the photon detector itself. This completes the triad of components, input, process and measurement, that are essential ingredients for any quantum information protocol \cite{Gisin10, Kimble08, Gisin07} and provides crucial information for the development of techniques to combat both imperfections \cite{Datta11-1, Hamilton09} and restrictions of experimental devices. Given these advancements it is therefore important to quantify how entanglement , an essential resource that enables quantum improvements over classical information protocols \cite{Horodecki09, Plenio07}, is affected by imperfections and restrictions on the available measurements. 

In quantifying the quantum information resources we usually proceed by fixing a certain set of dynamical processes that are allowed during this procedure \cite{Bartlett07}. The allowed processes are usually cheap, while those that are not allowed are expensive. Without this distinction, it is impossible to consider some states as resources. To see why this is the case, consider the reverse scenario - all dynamical processes are free or at least cheap. Given any two quantum states $\rho$ and $\sigma$, there exists some dynamical process that maps $\rho$ to $\sigma$. Therefore it does not matter which state we are given, as we can easily transform it to any other state, meaning that no state is particularly more useful than any other.

However, some states do become more important when a restriction is made concerning the allowed dynamical processes. As an example, suppose that a state $\rho$ is shared among two observers who are located very far away from one another and where transferring of one observer's part of the state to the other observer is very expensive, because quantum states are very fragile and difficult to transfer. Then the operations that are now cheap are precisely those that can be performed locally by the observers, allowing only classical communication between them. This defines the now ubiquitously used class of local operations and classical communication (LOCC). It has been shown that some states, such as separable states, cannot be transformed into all other states using only LOCC. This establishes a natural hieararchy of states, where those states that can be transformed into other states are higher up on the hierarchy and vice versa. The ability of the superior states to transfer into a greater number of states makes them a valuable resource, while the value of the states further down the chain is diminished. We therefore see that in order to understand the value of a given state, it is essential to understand which dynamical processes are readily available.

The states that are found closer to the top of the hierarchy are also better at performing those quantum information tasks that use only the operations from the allowed class. This is because the ability to transform the state $\rho$ into $\sigma$ implies that the state $\rho$ can be used to perform the task at least as well as the state $\sigma$ by first transforming it into $\sigma$ and then completing the task. This suggests another way to attain a hierarchy of states - compare their performance using a specific set of tasks one would like to use them for and place the states that perform better further up in the hierarchy structure.

However, in an experimental setup it is rather rare that one can peform all the operations in the often used classes of operations (such as LOCC). Almost always the apparatus will posses some kind of imperfections, such as for example making errors during measurement \cite{Braginski92}. In practice this imposes a further restriction on which dynamical processes one is able to perform. In our example, those processes that never make an experimental error cannot be performed. Since the allowed set of dynamical processes is therefore now smaller, we must correspondingly update our state hierarchy and potentially change the way resources are counted. Particularly, one would like to compare  the resource with restricted operations to that attainable without the additional restrictions. One potential method, and one we will adopt in this article, is to compare the performance of states in certain tasks with and without restrictions and then use the results to place the state at an appropriate position in the hierarchy. 

It has recently been discovered that the nonlocality and entanglement are intimately related in that all entangled states can be used to play semiquantum nonlocal games better than any separable state \cite{Buscemi12}. The games are constructed so that a referee sends a question in the form of a quantum state and the players, without communicating, perform joint measurements on the question state and their joint resource state. They receive payout depending on their answers, which in turn depend on the results of their measurements. It was shown in \cite{Buscemi12} that a state $\rho$ can generate higher payout with all nonlocal games than $\sigma$ if and only if $\rho$ can be tranformed into $\sigma$ using local operations and shared randomness (LOSR). No classical communication is allowed, as classical communication is a precious resource when it comes to nonlocality. Importantly, since LOSR operations cannot increase entanglement, we can conclude that higher payout with the state $\rho$ than $\sigma$ implies that $E(\rho) \geq E(\sigma)$, for any entanglement measure $E$ that is non-increasing under LOSR. Nonlocal games can therefore be used to witness and compare the amount of entanglement present in pairs of quantum states.
 
This article is organised as follows. In section \ref{sec:Definition} we formally define effective entanglement and prove several results characterising measures of it. In section \ref{sec:ImperfectM} we deal with the important case where the imperfection or restriction is described by a completely positive map (CPM)~\cite{Choi75} and show that effective concurrence reduces to being proportional to conventional concurrence~\cite{Wootters98,Gour05} with a CPM dependent scale factor. We then apply these results to quantify the entanglement of indistinguishable particles in section \ref{sec:Indistinguishable}, concentrating on the controversial concept of single-particle entanglement~\cite{Tan91,Hardy94, Peres95}, for the important examples of imperfect and super-selection rule restricted measurements. Finally we conclude in section \ref{sec:Conclusion}.

\section{Nonlocal games as gauge of effective entanglement}\label{sec:Definition}

The definition of effective entanglement relies heavily on using the semiquantum nonlocal games as the comparison gauge. We therefore first proceed to describe the rules of semiquantum nonlocal games (also see \cite{Buscemi12} for further details). They consist of four index sets, $\spc{S} = \{s\}$, $\spc{T} = \{t\}$, $\spc{X}=\{x\}$ and $\spc{Y} = \{y\}$. The referee picks indices $s$ and $t$ at random with probabilities $p(s)$ and $q(t)$ and prepares some corresponding quantum states $\zeta^s$ and $\eta^t$, sending them to players Alice and Bob, respectively. States corresponding to different indices need not be orthogonal, which is why the game is called a semiquantum nonlocal game, to distinguish it from a more classical nonlocal game where the states corresponding to different indices are fully distinguishable. The players must separately compute the respective answers $x \in \spc{X}$ and $y \in \spc{Y}$ and send them to the referee who then computes the payoff using the payoff function $\mathfrak{p}(s,t,x,y)$. Payoff need not be positive, in which case the players must pay the referee.  

Before the game begins the players may confer with one another and use any resources they like to do so in order to coordinate the strategy. After the game has begun, however, they are not allowed to communicate. All they can do is share a joint quantum state $\rho$ and perform joint measurements, described by positive operator-valued measure (POVM), on $\rho$ and the question states sent by the referee with outcomes in $\spc{X}$ and $\spc{Y}$, respectively. The average payoff Alice and Bob expect to obtain is expressed by the formula
\begin{align}
  \pay{\rho} = \max \sum_{s,t,x,y} p(s) q(t) \mathfrak{p}(s,t,x,y) \mu(x,y | s,t), \label{eq:MaximumAveragePayoff}
\end{align}
where $\mu(x,y | s,t)$ is the joint conditional probability of obtaining outcomes $x,y$ given that the question states $\zeta^s$ and $\eta^t$ were sent and is computed using the standard quantum probability formulae, the maximization is over all possible joint POVMS and the function $\pay{\rho}$ implicitly depends on the game players are playing. 

We call a state $\rho_1$ more nonlocal than $\rho_2$, denoted as $\rho_1 \succeq \rho_2$, if and only if for every semiquantum nonlocal game $\pay{\rho_1} \geq \pay{\rho_2}$. It is then shown in \cite{Buscemi12} that $\rho_1 \succeq \rho_2$ if and only if $\rho_1$ can be transformed to $\rho_2$ using only LOSR. This is denoted as $\rho_1 \mapsto \rho_2$. Given an entanglement measure $E$ that is non-increasing under LOSR, we then have that 
\begin{align}
  \rho_1 \succeq \rho_2 \Rightarrow E(\rho_1) \geq E(\rho_2). \label{eq:EntanglementGauge}
\end{align}
Notice that this implies that whenever $\pay{\rho_1} = \pay{\rho_2}$ for all games we must have that $E(\rho_1) = E(\rho_2)$, justifying the idea that semiquantum nonlocal games act as an entanglement gauge. LOSR transformations are a subset of LOCC operations and so any entanglement measure satisfying the standard required properties can be used (for a review of the properties of entanglement measures, see for instance \cite{Plenio07, Horodecki09}). 

Before we look at the nonlocal games in the presence of POVM restrictions, we define what we mean by an effective POVM. Suppose $\op{P}_1, \ldots \op{P}_2$ is a joint POVM over two Hilbert spaces, with the joint state being a product state $\rho_1 \otimes \rho_2$. Then there exists a $\rho_1$-dependent POVM acting only on $\rho_2$ with the same outcome probabilities. Writing $\op{P}_k = \sum_l \lambda_{kl} \op{P}^1_{kl} \otimes \op{P}^2_{kl}$, we can obtain the effective POVM by computing $\Tr_1[\op{P}_k \rho_1 \otimes \rho_2] = \sum_l \lambda_{kl} \Tr[\op{P}_{kl}^1 \rho_1] \op{P}_{kl}^2 \rho_2$, giving $\sum_l \lambda_{kl} \Tr[\op{P}_{kl}^1 \rho_1] \op{P}_{kl}^2$ as the effective operator. 

Since we want to look at the effective entanglement when allowed POVMs acting on $\rho$ alone are restricted, we therefore maximize the Eq.~\ref{eq:MaximumAveragePayoff} over only those POVMs, whose effective POVM on $\rho$ alone belongs to a certain set of POVMs $\spc{R}$. We denote this new maximum average payoff as $\payr{\rho}$. Denote as $\spc{E}$ the set of all states $\bar{\sigma}$ such that $\pay{\sigma} \geq \payr{\rho}$. In words, these are those states whose maximum average payoff function with no restrictions on POVMs is at least as great as the maximum payoff function of state $\rho$ with POVMs restricted. We then define the effective entanglement as
\begin{align}
  \bar{E}(\rho) = \inf\{E(\bar{\sigma}): \bar{\sigma} \in \spc{E}\}.
\end{align}
The functional $\bar{E}(\rho)$ therefore gives us the least amount of entanglement that enables us to perform at least as well in any nonlocal game with no restrictions as we could with the more entangled state $\rho$ with restrictions. From the fact that $\rho \in \spc{E}$ it follows that $\bar{E}(\rho) \leq E(\rho)$. The infimum can be replaced by minimum whenever the set $\spc{E}$ is compact and the entanglement measure $E$ is continuous. This follows from a fundamental theorem of classical analysis, saying that continuous functions mapping compact sets to real numbers attain their infimum on a member of the set \cite{CarothersAnalysis, RudinAnalysis}. Since $\spc{E} \subset \spc{B}(\spc{H})$ and $ \spc{B}(\spc{H})$ is compact, $\spc{E}$ is also compact whenever it is closed due to the fact that closed subsets of compact sets are compact. 

It is useful to have a condition that allows us to more easily find and verify when a particular functional is the effective entanglement $\bar{E}$. We now provide one such condition.
\begin{theorem}
  If there exists a state $\bar{\rho} \in \spc{E}$ such that for all semiquantum nonlocal games $\pay{\bar{\rho}} = \payr{\rho}$ then $\bar{E}(\rho) = E(\bar{\rho})$. 
\end{theorem}
\begin{proof}
  To see this notice that for any $\bar{\sigma} \in \spc{E}$ we have that $\pay{\bar{\sigma}} \geq \payr{\rho} = \pay{\bar{\rho}}$ for all semiquantum nonlocal games. Therefore, $\bar{\sigma} \succeq \bar{\rho}$ and $E(\bar{\sigma}) \geq E(\bar{\rho})$ by the implication \eqref{eq:EntanglementGauge}. Therefore $\bar{E}(\rho) = E(\bar{\rho})$. 
\end{proof}
The theorem is saying that when there is a state $\bar{\rho}$ that is exactly as good for playing nonlocal games with unrestricted measurements as the state $\rho$ is with restricted measurements, this state can be considered to be the effective state that gives us the effective entanglement. Given that precisely equal performance in nonlocal games implies equal entanglement, this theorem gives a property that effective entanglement should naturally be expected to satisfy.

In the following section we will consider a particularly simple description of measurement restrictions in terms of completely positive maps, encompassing a large range of measurement errors found in modern quantum laboratories. We will provide the exact formulae for computing effective entanglement under such circumstances and show that effective G-concurrence, generalization of concurrence \cite{Gour05}, is proportional to the usual G-concurrence when measurement errors occur for only one of the observers. When they occur for both, the proportionality gives us a simple to compute upper bound. 

\section{Imperfections and restrictions via CPMs}\label{sec:ImperfectM}

To define the restricted set $\spc{R}$ of effective POVMs we are able to measure, we require that any POVM $\{\op{P}_k\}$ in $\spc{R}$ is the image of some other completely arbitrary and general POVM $\{\op{G}_k\}$ under the action of some local CPM $\$ = \$_A \otimes \$_B$ so that
\begin{align}
  \op{G}_k = \$_A^\dagger \otimes \$_B^\dagger[\op{P}_k]. \label{eq:POVMRelation}
\end{align}
Notice that this is equivalent to the CPM $\op{\$}$ acting on the state $\rho$, since
\begin{multline}
p_k = \Tr\left(\$^\dagger[\op{P}_k]\rho \right) = \Tr\left( \sum_j \op{K}_j \op{P}_k \op{K}_j^\dagger \rho \right)  \\ = \Tr\left( \op{P}_k \sum_j \op{K}_j^\dagger \rho \op{K}_j \right) = \Tr\left( \op{P}_k \$\left[\rho\right]\right), \label{eq:ImperfectStateProof}
\end{multline}
where $\op{K}_j$ are the Kraus operators associated with the CPM $\$_A \otimes \$_B$. The Eq.~\eqref{eq:POVMRelation} is thus the Heisenberg picture equivalent of the operation $\$_A \otimes \$_B$ acting on the state $\rho$. It is therefore unsurprising that the state $\$_A \otimes \$_B(\rho)$ is the effective state $\bar{\rho}$ that can be used to define effective entanglement. 

\begin{figure}[t]
\centering
\includegraphics[width=7cm]{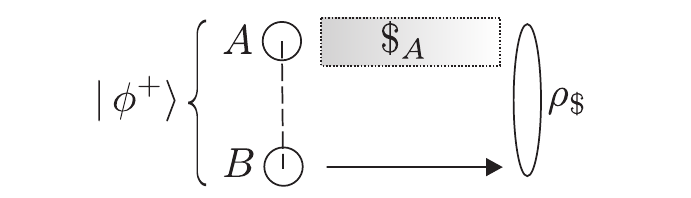}
\caption{(a) In this work we model an imperfect or restricted detector as an ideal detector augmented with a CPM channel $\$$ on its input. (b) A depiction of the state-channel isomorphism. Here a channel $\$$ acts on one part of the maximally entangled bipartite state $\ket{\phi^+}$. The resulting density matrix $\rho_\$$ is the state dual to the channel and completely characterises it. The map between the two is an isomorphism so it is possible to invert it and obtain the channel $\$$ back from the density matrix $\rho_\$$.}
\label{fig:ChannelStateDuality}
\end{figure}

The concept of using $\bar{\rho}$ to compute effective entanglement becomes particularly appealing when the entanglement measure $E$ used is G-concurrence in $d$ dimensions $G_d$. This measure is an entanglement monotone for all $d$ and reduces to concurrence when $d=2$ and is, in that case, therefore explicitly related to the entanglement of formation $E_F$ \cite{Gour05, Wootters98}. Furthermore, it is shown in \cite{Gour05} that all bipartite entanglement monotones can be written in terms of G-concurrence.  For arbitrary pure states it is given by $G_d(\ket{\psi}\bra{\psi}) = d \left(\det[A^\dagger A]\right)^{1/d}$ where $\ket{\psi} = \sum_{i,j} A_{i,j} \ket{i}\otimes\ket{j}$, while for mixed states in general it is given by the convex roof extension
\begin{align}
    G_d(\rho) = \inf\left\{ \sum_i p_i G_d\left(\ket{\psi_i}\bra{\psi_i}\right) \right\}, \label{eq:ConvexRoof}
\end{align}
where infimum is over all convex decompositions $\rho = \sum_i p_i \ket{\psi_i}\bra{\psi_i}$, with the states $\ket{\psi_i}$ not necessarily being orthogonal. Eq.~\eqref{eq:ConvexRoof} automatically ensures that the measure is convex, which is a very convenient property for an entanglement measure to posses \cite{Plenio07}. 

Recent results on entanglement evolution \cite{Konrad08, Tiersch08, Tiersch09, Gour10} under local CPMs give us a particularly simple expression for effective G-concurrence when the operation $\$ = \$_A \otimes \op{1}$ or $\$ = \op{1} \otimes \$_B$. In this case we have that
\begin{align}
  \bar{G}_d \left(\ket{\psi}\bra{\psi}\right) = Q(\$_A) G_d\left(\ket{\psi}\bra{\psi}\right),
\end{align}
where $0 \leq Q(\$_A) \leq 1$ is a quality factor given by $Q(\$_A) = G_d\bigl[(\$_A \otimes \op{1})(\ket{\phi_d}\bra{\phi_d})\bigr]$ and $\ket{\phi_d} = \sum_{k=0}^{d-1} \ket{kk}/\sqrt{d}$ is a maximally entangled state in $d\times d$ dimensional Hilbert space. From convexity of $G_d$ it follows that
\begin{align}
  \bar{G}_d \left(\rho\right) & \leq Q(\$_A) G_d\left(\rho\right).
\end{align}
Similarly, when $\$ = \$_A \otimes \$_B$ and also due to convexity,
\begin{align}
  \bar{G}_d \left(\rho\right) \leq Q(\$_A)Q(\$_B) G_d\left(\rho\right).
\end{align}
Thus we have a simple linear relationship between the effective G-concurrence and the conventional G-concurrence, with the proportionality factor given by the quality $Q$, which is a function of the restriction $\$$ only. The density matrix $\rho_{\$} = \op{1} \otimes \$_B [\ket{\phi_d}\bra{\phi_d}]$, shown in Fig.~\ref{fig:ChannelStateDuality}(b) follows from the state-channel isomorphism~\cite{Hein05, NielsenChuang, GeometryQuantum} and is known to completely characterise the channel $\$$. An extreme case is where $\rho_{\$}$ is completely separable so $Q(\$)=0$. Such a channel is said to be entanglement breaking \cite{Hein05} since the output state of the channel is separable for any input state. In this context it implies that effective entanglement vanishes. The linearity makes the effective entanglement, or its upper bound, particularly simple to compute when G-concurrence is used as the measure of entanglement, and we do so analytically for several examples we expect appear particularly often in laboratory settings. 

\section{Entanglement of indistinguishable particles}\label{sec:Indistinguishable}
The underlying physical framework for entanglement is a system composed of two or more individually addressable degrees of freedom that together form a tensor product Hilbert space where the state of the complete system is described. The archetypal case is $\mathbbm{C}^2 \otimes \mathbbm{C}^2$ for two qubits, which is usually envisaged as arising from two localized particles with spin-$\frac{1}{2}$. Conceptually entanglement is then signalled by a lack of separability of the state of the system with respect to this tensor product structure. However, if particles are delocalized and indistinguishable this raises a significant issues for quantifying entanglement because the relevant degrees of freedom can no longer be assigned to individual particles. Instead an analysis of entanglement requires a description in terms of the second quantized field modes which the particles can occupy. The most elementary example of this problem consists of a single-particle delocalized across two distinct spatial modes $\op{a}$ and $\op{b}$ yielding a state~\cite{Tan91,Hardy94, Peres95}
\begin{align}
\ket{\Psi^\pm} = \frac{1}{\sqrt{2}}(\op{a}^\dagger + \op{b}^\dagger)\ket{\textrm{vac}} = \frac{1}{\sqrt{2}}\left(\ket{0}_a\ket{1}_b \pm \ket{1}_a\ket{0}_b\right) \nonumber,
\end{align}
where $\ket{n}_a\ket{m}_b \propto (\op{a}^\dagger)^n(\op{b}^\dagger)^m\ket{\textrm{vac}}$. Interpreting the Fock states $\ket{0}$ and $\ket{1}$ of either mode as Pauli $\sigma_z$ eigenstates of a qubit suggests that the second quantized form $\ket{\Psi^\pm}$ it is an entangled state. Yet this notion that a single-particle can truly exhibit entanglement has raised considerable controversy~\cite{Greenberger95,Vaidman95,Pawlowski06,Drezet06,vanEnk05}. This issue is quite naturally analysed within the framework of effective entanglement since the accessibility of any correlations in states such as $\ket{\Psi^\pm}$ is fundamentally linked to what measurements are available. Indeed a central issue is whether full discrimination of the Bell basis is possible, namely if $\ket{\Psi^\pm}$ and
\begin{align}
\ket{\Phi^\pm} = \frac{1}{\sqrt{2}}(1+ \op{a}^\dagger \op{b}^\dagger)\ket{\textrm{vac}} = \frac{1}{\sqrt{2}}\left(\ket{0}_a\ket{0}_b \pm \ket{1}_a\ket{1}_b\right), \nonumber
\end{align}
can be measured. This is essential for such states to be a resource in protocols like teleportation \cite{Lombardi02,Lutkenhaus99}. Recently the Bell discrimination has been shown to be possible with photons when non-linear optics are used in combination with a two level atom \cite{Bjork12}. In the following we study single-particle entanglement for the case of photons with imperfect detectors and for massive particles where a super-selection rule physically restricts measurements.

\subsection{Optical amplitude and phase damping} 
The case of photons provides a simple test ground for effective entanglement. Suppose that the state $\ket{\Psi^\pm}$ was used within an LOCC protocol where Alice implements imperfect measurements of her optical mode $\op{a}$. With a photon counter she can measure $\op{a}^\dagger\op{a}$, or via a balanced homodyne detector with a local oscillator with a phase $\phi$ she can measure a field quadrature $X(\phi) = (\op{a} \,e^{-i\phi} + \op{a}^\dagger e^{-i\phi})/2$. Imperfections in photon counting might for example cause the detector not to `click'  due to photon loss within the device. This is amplitude damping and can be modelled as a beam-splitter at the input port of a perfect counter which scatters an incoming photon into another unmonitored optical mode, as shown in Fig.~\ref{fig:photon_detect}(a). Imperfections in the measurements of field quadratures might arise due to uncertainty in the phase $\phi$. This is phase damping and can be modelled by a local oscillator subject to phase fluctuations, as shown in as shown in Fig.~\ref{fig:photon_detect}(b). For both types of measurements their errors, within the subspace where no more than one photon occupies the local mode, are described by a CPM of the form $\$[\rho] = \op{E}_0 \rho \op{E}_0^\dagger + \op{E}_1 \rho \op{E}_1^\dagger$. Amplitude damping has $\op{E}_0 = \ket{0}\bra{0} + \sqrt{1-\gamma}\ket{1}\bra{1}$ and $\op{E}_1 = \sqrt{\gamma} \ket{0}\bra{1}$, with a photon loss rate $\gamma$. Phase damping has $\op{E}_0 = \ket{0}\bra{0} + \sqrt{1-\lambda}\ket{1}\bra{1}$ and $\op{E}_1 = \sqrt{\lambda} \ket{1}\bra{1}$, with phase flipping elastic photon scattering occurring at a rate $(1-\sqrt{1-\lambda})/2$. In either case we can characterise the imperfect measurements via the state-channel isomorphism using $\ket{\Psi^\pm}$ by computing $\rho_\$ = (\$ \otimes \op{1})\ket{\Psi^\pm}\bra{\Psi^\pm}$. The concurrence of $\rho_\$$ then gives the proportionality factor $Q[\$]$ between effective concurrence and the standard concurrence of any pure state as $Q[\$] =  \sqrt{1-\gamma}$ and $Q[\$] =  \sqrt{1-\lambda}$ for amplitude and phase damping, respectively. As expected these imperfections monotonically erode the effective entanglement accessible within a state such as $\ket{\Psi^\pm}$.

\begin{figure}[t]
\centering
\includegraphics[width=7cm]{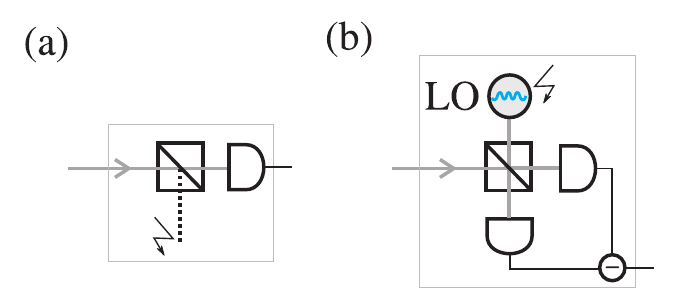}
\caption{(a) A photon counter for the input mode suffering from amplitude damping noise of rate $\gamma$. The photon loss within the device is modelled by a beam-splitter with reflectivity $\gamma$. (b) A balanced homodyne device for measuring the field quadrature $X(\phi)$ of the input mode, where $\phi$ is controlled by the phase of the local oscillator (LO). Phase fluctuations of the LO cause phase damping noise in the device with a rate $\lambda$. The relation between the phase fluctuations and $\lambda$ is identical to those discussed later in Sec.~\ref{sec:massive}.}
\label{fig:photon_detect}
\end{figure}

\subsection{Massive particles subject to super-selection rules} \label{sec:massive}
If the modes $\op{a}$ and $\op{b}$ correspond to those of a massive particle then, in contrast to photons, super-selection rules impose a fundamental physical restriction on both the operations and measurements that can be performed. Specifically, massive particles are subject to Bargmann's super-selection rule for non-relativistic quantum mechanics \cite{Bargmann54, Cisneros98} which prohibits superpositions of states of different mass like those seen in the Bell states $\ket{\Phi^\pm}$. Strict adherence to this super-selection rule also means that the only permissible local measurements are those which commute with the number operator $\op{a}^\dagger\op{a}$ meaning also that measurements of local superpositions between $\ket{0}$ and $\ket{1}$ are inaccessible. With these restrictions one might reasonably question whether a single massive particle delocalized over two spatial regions in a state $\ket{\Psi^\pm}$ is really entangled. This issue has been hotly debated recently \cite{Ashhab09, Heaney09, Cunha07}.

In fact if a broken-symmetry reference frame is present it can partially or fully lift the super-selection rule \cite{Bartlett07, Paterek11, Vaccaro08} allowing coherences to be measured. This has sparked investigations into both the quantification of entanglement in such scenarios \cite{Bartlett06-1, Bartlett03-1, Schuch04} as well as the potential presence of single-particle quantum nonlocality \cite{Schuch04, Heaney10, Ashhab07-1}. As a result incorporating full or partial super-selection rule measurement restrictions into the quantification of entanglement is a fundamental requirement for building a meaningful entanglement measure for such systems \cite{Wiseman03, Jones06}. The framework of effective entanglement introduced here provides an exemplary tool in this context. Specifically, while the effective entanglement in the extreme cases of no restrictions and complete adherence to super-selection rules is currently understood, the intermediate cases permitted by a general broken-symmetry reference frame are not. 

In the case where reference frames are available which fully lift the super-selection rule restrictions then standard entanglement measures are sufficient. In the opposite case where no reference frames are present, the super-selection rule restrictions are described by a CPM, just like in the previous section. Since no coherences between different particle number sectors can be measured, the CPM must remove them from the measurement POVMs. Such CPM would take the form
\begin{align}
    \$\left[\rho\right] = \sum_n \Pi_n \rho \Pi_n,
\end{align}
where $\Pi_n$ is the projector onto the subspace of $n$ particles and $\rho$ is some state with fixed total number of particles. It is therefore sufficient for the above CPM to act on only one of the party's subsystems. If both Alice and Bob are affected by the same restriction, then we can get the effective state $\$[\rho] = \bar{\rho} = \bigoplus_n p_n \rho_n$ where the direct sum $\oplus$ signifies that it has block-diagonal form where $\rho_n = \Pi_n \rho \Pi_n$ and $p_n = \Tr\left[\Pi_n \rho \right]$. Thus, super-selection rule restricted effective entanglement measure for pure $\rho$ is given by
\begin{align}
	E(\bar{\rho}) = \sum_n p_n E(\rho_n),
\end{align}
where $E$ entanglement of formation. For a mixed state the above forms an upper bound. In fact this measure of entanglement for indistinguishable particles was introduced already by Wiseman and Vaccaro ~\cite{Wiseman03}. Here we note that the framework of effective entanglement has led naturally to the same result. 

\begin{figure}[t]
\centering
\includegraphics[width=8cm]{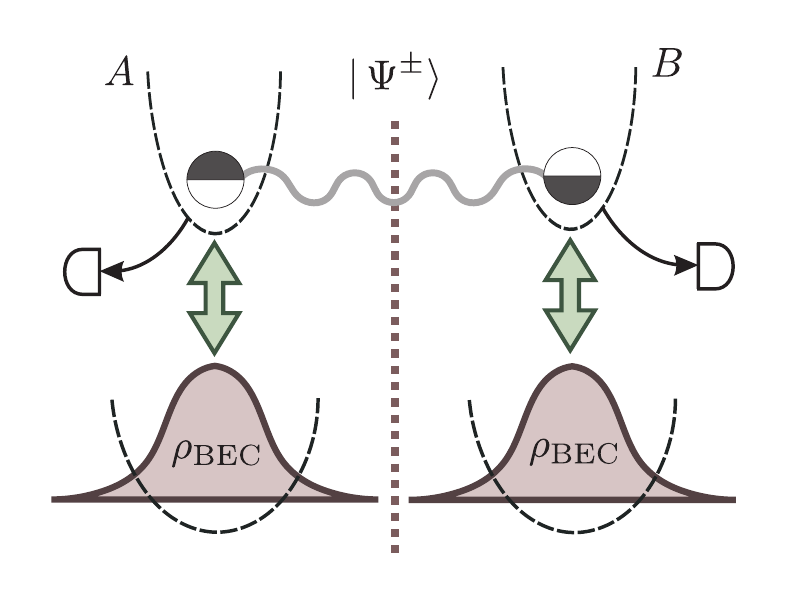}
\caption{Two parties Alice $(A)$ and Bob $(B)$ share a single-particle entangled state $\ket{\Psi^\pm}$ whose superposition between two tight potentials is depicted as the half-filled circles. For her part of in an LOCC protocol $\Lambda$ Alice is required to measure a superposition of the particle number states $\ket{0}$ and $\ket{1}$ of her local mode $\hat{a}$. This coherent rotation is achieved by interacting her mode with an ancilla mode $\hat{c}$ which form a local BEC reference frame for a certain amount of time. She then measures the occupation of her local mode with an ideal detector. If the BEC has a large occupation and well defined phase then any rotation can be performed perfectly, breaking the super-selection rule restriction. In contrast if the local BEC has a completely uncertain phase then the rotation is completely incoherent and the super-selection rule restriction remains.}
\label{fig:BECDrawing}
\end{figure}

A cold-atom inspired setup provides a concrete example for exploring the general case in between these limits. Specifically, the resource state $\ket{\Psi^\pm}$ now describes the state of a single atom in a superposition over two tightly confined potentials, as implemented by two atomic quantum dots where the repulsive interactions between atoms is sufficiently large to prohibit double occupancy \cite{Bruderer06}. To exploit the resource state in some LOCC protocol $\Lambda$ Alice may be required to measure a super-selection rule violating superposition of the particle number states $\ket{0}$ and $\ket{1}$ of her local mode $\hat{a}$. She therefore needs access to a unitary transformation of the mode equivalent to the single-qubit rotation about the $x$-axis $\hat{R}_x(\theta) = \exp(-i\theta\sigma_x/2)$, where $\sigma_x$ is the Pauli $x$ operator\footnote{Note that the mode equivalent of a single-qubit rotation about the $z$-axis $\hat{R}(\vartheta) = \exp(-i\vartheta\sigma_z/2)$ can be trivially implemented via evolution of the mode with a Hamiltonian of the form $\kappa\op{a}^\dagger\op{a}$.}. To implement such a rotation she exploits a local reference frame composed of an ancilla mode $\op{c}$ in a Bose Einstein condensate (BEC) described by a mixture of coherent states 
\begin{align}
   \rho_{\textrm{BEC}} = \int_0^{2\pi} d\phi p(\phi) \ket{|\alpha| e^{i \phi}}\bra{|\alpha| e^{i \phi}}, \label{eq:rho_bec}
\end{align}
where $\ket{\alpha} = \exp(-|\alpha|^2/2) \sum_{n=0}^\infty (\alpha \op{c}^\dagger)^n\ket{\textrm{vac}}/n!$ with $\alpha$ a complex number and $p(\phi)$ is the phase distribution.  Owing to the global phase being unobservable the phase distribution $p(\phi)$ and its translation $p(\phi+\phi_0)$ create physically equivalent states $\rho_{\textrm{BEC}}$. Otherwise, depending on the structure of $p(\phi)$ the reference frame state $\rho_{\textrm{BEC}}$ either breaks or adheres to the number symmetry underlying the super-selection rule.   

The execution of the $R_x(\theta)$ rotation is then attempted by Alice jointly evolve her local mode $\op{a}$ and the ancilla mode $\op{c}$ via the Hamiltonian $\op{H} = -\frac{1}{2}\Omega \left(\op{a}^\dagger \op{c} + \op{c}^\dagger \op{a}\right)$ for a given time $t$. This number-symmetric interaction drives an exchange of particles between the BEC reservoir and the resource state, after which the BEC reservoir is traced out. Such a measurement is depicted in Fig. \ref{fig:BECDrawing}. The effect of this evolution is best revealed by analysing one coherent state $\ket{|\alpha| e^{i \phi}}$ in the mixture $\rho_{\textrm{BEC}}$. In the limit $|\alpha|^2 \gg 1$ this gives\footnote{We write $\rightarrow$ instead of $\approx$ since exact evolution converges extremely rapidly as a function of $|\alpha|^2$ to the product form mapping shown.}
\begin{eqnarray}
	\ket{0}\ket{|\alpha| e^{i \phi}}  &\rightarrow& \left(\cos(\omega t) \ket{0} -i e^{i\phi}\sin(\omega t) \ket{1}\right)\ket{|\alpha| e^{i \phi}}, \nonumber \\
	\ket{1}\ket{|\alpha| e^{i \phi}}  &\rightarrow& \left(\cos(\omega t) \ket{1} -i e^{i\phi}\sin(\omega t) \ket{0}\right)\ket{|\alpha| e^{i \phi}} \nonumber,
\end{eqnarray}
where $\omega = \frac{1}{2}\Omega |\alpha|$. This represents a unitary evolution of the mode equivalent to the sequence of rotations $\hat{R}_z(\phi) \hat{R}_x(\omega t)$. Tracing the BEC out then yields an effective evolution of the mode as
 \begin{align}
   \Gamma[\rho] = \int_0^{2\pi} d\phi p(\phi) \hat{R}_z(\phi) \hat{R}_x(\omega t) \rho \hat{R}^\dagger_x(\omega t) \hat{R}^\dagger_z(\phi).
\end{align}
This is equivalent to applying the unitary $\hat{R}_x(\omega t)$ to the input state as $\rho_R = \hat{R}_x(\omega t) \rho \hat{R}^\dagger_x(\omega t)$ followed by a phase-damping channel so
\begin{align}
	\Gamma[\rho] = (1-|g|) \left(\proj{P}_0 \rho_R \proj{P}_0 + \proj{P}_1 \rho_R \proj{P}_1\right) + |g| \rho_R,
\end{align}
where $g = -i \int_0^{2\pi} d\phi p(\phi) \exp(i \phi)$, $\proj{P}_0 = \ket{0}\bra{0}$ and $\proj{P}_1 = \ket{1}\bra{1}$. 

The CPM $\Gamma$ describing this measurement imperfection is not yet in the model form we restricted to in Sec.~\ref{sec:ImperfectM}. Instead we identify the necessary map $\$$ as
\begin{align}
	\$[\rho] = \hat{R}^\dagger_x(\omega t)\Gamma[\rho]\hat{R}_x(\omega t).
\end{align}
The CPM factor in the case is $Q[\$] = |g|$, identical to a phase damping channel with a rate $\sqrt{1-\lambda} = |g|$, since the rotations $\hat{R}_x(\omega t)$ have no influence on entanglement. Thus $|g|=0$ is a sufficient condition for vanishing effective entanglement and for pure input states it is also necessary.

If the BEC possesses a well defined phase $\phi_0$ so $p(\phi) = 2\pi \delta(\phi - \phi_0)$ then $|g|=1$ and $\Gamma$ reduces to the unitary $\hat{R}_z(\phi_0) \hat{R}_x(\omega t)$ and Alice succeeds in rotating the state of mode $\op{a}$ into the desired pure coherent superposition of $\ket{0}$ and $\ket{1}$. In this case the super-selection rule is fully lifted and there are no restrictions on what can be measured \cite{Bartlett07, Heaney09-1}. This highlights that the BEC plays the role of a perfect local oscillator analogous to homodyne detection of field quadratures. In contrast, a completely uncertain phase $p(\phi) = (1/2\pi)$ gives $|g| = 0$ so the attempted rotation can only generate statistical mixtures of $\ket{0}$ and $\ket{1}$ in strict adherence to the super-selection rule. 

\begin{figure}[t]
  \centering
  \includegraphics[width=7cm]{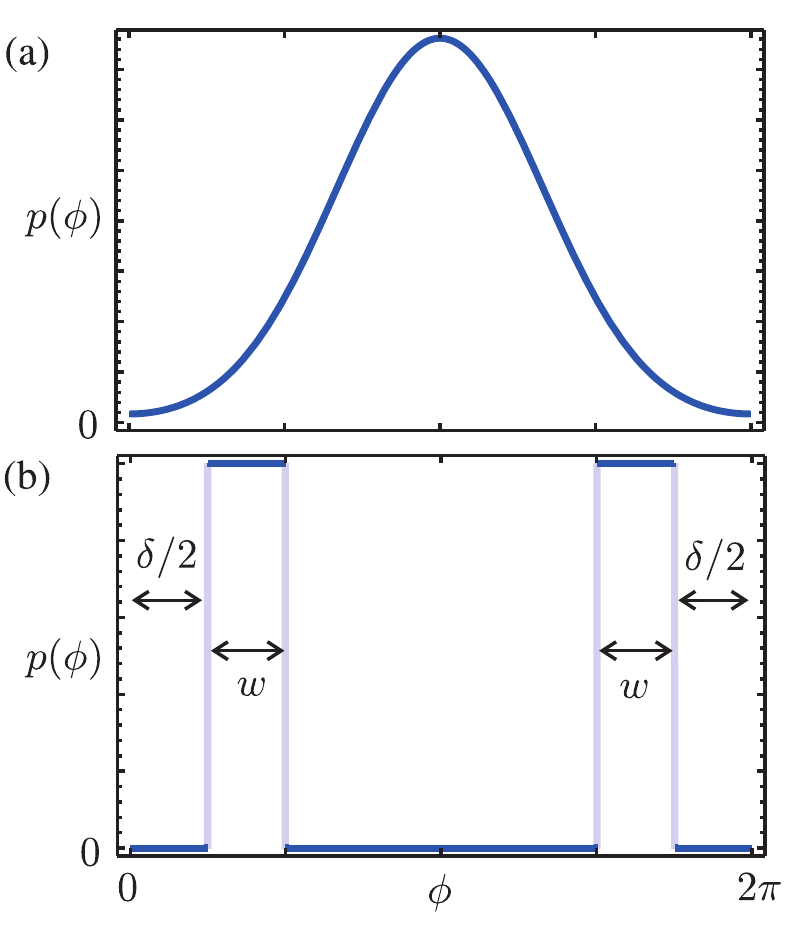}
  \caption{Two different examples of incoherent phase distributions $p(\phi)$ of the BEC reference frame. In (a) a wrapped normal distribution is shown which results in $|g|=\exp(-\sigma^2/2)$, where $\sigma$ is variance. In (b) sectionally constant distribution is shown which gives $|g| = (4w/\pi) \sin(w/2)\cos(w/2+\delta/2)$, where $w,\delta$ are as shown.}
  \label{fig:BECGraphs}
\end{figure}

For general distributions $p(\phi)$ allow $|g|$ to vary between these limiting cases yield a partial lifting of the super-selection rule restrictions. In this case it is not possible to coherently evolve into any superposition of $\ket{0}$ and $\ket{1}$ without at least some degree of mixing determined by $|g|$ Since $g$ is the average value of $\exp(i \phi)$, its absolute value is a direct measure of the amount of the reliability of the BEC as a phase reference. To illustrate this we computed $|g|$ for a Gaussian distribution wrapped around a circle
\begin{align}
  p(\phi) = \frac{1}{\sigma \sqrt{2\pi}} \sum_{k=-\infty}^{\infty} \exp\left[ \frac{-(\phi-\mu+ 2\pi k)^2}{2\sigma^2}\right],
\end{align}
as shown in Fig. \ref{fig:BECGraphs}(a). We get $|g| = \exp(-\sigma^2/2)$ by integrating its uniformly convergent series over the interval $\phi \in [0,2\pi]$ term by term. As expected, effective entanglement decreases with increasing phase uncertainty $\sigma$. We also considered the sum of two flat distributions of width $w$ on the circle with centres shifted from $\phi = 0$ and $\phi = 2\pi$ by $\delta/2$, as depicted in Fig.~\ref{fig:BECGraphs}(b). Here we find that as long as there are no overlaps $|g| = (4w/\pi) \sin(w/2)\cos(\delta/2+w/2)$. 

The phase of the BEC is the crucial property allowing it to act as a super-selection rule breaking reference frame. Mixing a BEC of phase $\phi$ with one of phase $\phi+\pi$, by forming a state such as
\begin{align}
  	\rho_{\textrm{BEC}} \propto \uket{|\alpha| e^{i \phi}}\ubra{|\alpha| e^{i \phi}}+  \uket{|\alpha| e^{i (\phi+\pi)}}\ubra{|\alpha| e^{i (\phi+\pi)}}, \nonumber
\end{align}
gives $|g|=0$. However, $\rho_{\textrm{BEC}}$ itself still contains coherences which violate the super-selection rule. This shows that while symmetry breaking is a necessary condition for having non-vanishing effective entanglement, it is not sufficient. 

\section{Conclusion}\label{sec:Conclusion}
We have introduced an effective entanglement functional, measuring the minimum amount of entanglement needed to perform semiquantum nonlocal games with perfect measurements at least as well as with the imperfect measurements. The semiquantum nonlocal games posses properties that make them ideal as a gauge of the amount of entanglement in a state, the property we make extensive use of. We showed that whenever we can describe the restrictions through one-sided CPMs, an exact result can be obtained where the effective G-concurrence is proportional to the conventional G-concurrence, with proportionality coefficient given by the CPM dependent quality factor. For two-sided CPMs and mixed states the expression gives an upper bound. Although we have only dealt with bipartite entanglement measures in this paper, multipartite measures could similarly be treated using analogous results for multipartite entanglement evolution given in \cite{Gour10}, again obtaining an emergent quality factor. 

We should note that although we used semiquantum nonlocal games as entanglement gauge here, other entanglement gauges could equally be applicable. The crucial property that enabled us to use semiquantum nonlocal games is that whenever the maximum average payoff function $\pay{\rho} \geq \pay{\sigma}$ for all semiquantum nonlocal games we also have that $E(\rho) \geq E(\sigma)$. Any set of quantum protocols where there is a POVM dependent fidelity such that $F(\rho) \geq F(\sigma)$ for all protocols in the set $\Rightarrow E(\rho) \geq E(\sigma)$ would provide equal results. Conversely, if higher entanglement $E(\rho) \geq E(\sigma)$ implies higher POVM dependent fidelity $F(\rho) \geq F(\sigma)$ for some set of protocols, then entanglement is a resource for these protocols. When measurements are restricted, the resource naturally becomes the effective entanglement defined in this paper $\bar{E}(\rho)$. 

We applied this framework to describe single-particle entanglement. For the case of photons we showed how effective entanglement is attenuated by common measurement noise like amplitude and phase damping. For massive particles we considered the fundamental restrictions imposed by super-selection rules. Computing effective entanglement for strict adherence of the rule we found the entanglement of particles \cite{Wiseman03}. We further extended this work by considering measurements that can utiliize a BEC with phase uncertainty as a reservoir to partially lift the super-selection rule restriction. The multiplicative factor for concurrence was found explicitly, from which we deduce that whenever the BEC has an undefined phase effective entangled vanishes implying that the same protocol could be performed without entanglement.

\acknowledgments
SM would like to thank Stephen R. Clark and Dieter Jaksch for many helpful discussions and the Engineering and Physical Sciences Research Council (EPSRC) for financial support.
\bibliography{Entanglement}

\end{document}